\documentclass[11pt,leqno,fleqn]{article}

\newcommand{\Comment}[1]{\relax}
\usepackage{amsthm}
\usepackage{latexsym}
\usepackage{amsmath}
\usepackage{amssymb}
\usepackage{amsfonts}
\usepackage{mathtools}
\usepackage{subcaption}
\usepackage{algorithm}
\usepackage{algorithmicx}
\usepackage{amsbsy}
\usepackage{synttree} 
\usepackage{diagrams}
\usepackage{enumitem} 
\usepackage[usenames]{color}
\usepackage[colorlinks=false, %
           citecolor=red,filecolor=blue,linkcolor=blue,urlcolor=blue,
           linktoc=page,
           pagebackref=false
           ]{hyperref} 

\usepackage{fullpage} 
\usepackage{afterpage}  
\usepackage{float} 
\usepackage{wrapfig}  
\floatstyle{boxed} 

\usepackage{times}
\usepackage{bm}
\usepackage{stmaryrd} 
\usepackage{MnSymbol}  
\usepackage{mathrsfs}  
\usepackage{graphicx}
\usepackage[all]{xy}
\usepackage{fancybox}
\usepackage{alltt}

\newcommand{\Hide}[1]{}

\newif
\ifnote
\notetrue

\newif
\ifTR
\TRfalse

\newtheorem{theorem}{Theorem}
\newtheorem{lemma}[theorem]{Lemma}

\newtheorem{restxxx}[theorem]{Restriction}

\newtheorem{agreexxx}[theorem]{Agreement}

\newtheorem{termxxx}[theorem]{Terminology}

\newtheorem{notxxx}[theorem]{Notation}

\newtheorem{assumxxx}[theorem]{Assumption}

\newtheorem{convenxxx}[theorem]{Convention}

\newtheorem{exaxxx}[theorem]{Example}
\newenvironment{example}{\begin{exaxxx}\rm}{\hfill\QED\end{exaxxx}}
\newtheorem{exexxx}[theorem]{Exercise}

\newtheorem{remxxx}[theorem]{Remark}

\newtheorem{openxxx}[theorem]{Open Problem}
\newtheorem{conjxxx}[theorem]{Conjecture}

\newtheorem{defxxx}[theorem]{Definition}
\newenvironment{definition}[1]{\begin{defxxx}[\emph{#1}]\rm}%
{\hfill\QED\end{defxxx}}
\newtheorem{procxxx}[theorem]{Procedure}
\newenvironment{procedure}[1]{\begin{procxxx}[\emph{#1}]\rm}%
{\hfill\QED\end{procxxx}}

\newtheorem{Prxxx}[theorem]{Proof}
{\end{Prxxx}} 

  {\addtolength{\leftskip}{#1}\addtolength{\rightskip}{#2}}{\par}

\newcommand{\Parens}[1]{\bigl(#1\bigl)} 
\newcommand{\Set}[1]{\{ #1 \}}

\newcommand{\bigO}[1]{{\cal O}\bigl(#1\bigr)} 

\newcommand{\Let}[3]%
    {\textbf{\textsf{let}}\ {#1}\,{#2}\ \textbf{\textsf{in}}\;{#3}\,}
\newcommand{\Try}[3]%
    {\textbf{\textsf{try}}\ {#1} {#2}\ \textbf{\textsf{in}}\;{#3}\;}
\newcommand{\Mix}[3]%
    {\textbf{\textsf{mix}}\ {#1} {#2}\ \textbf{\textsf{in}}\;{#3}\;}
\newcommand{\LET}[3]%
    {\textbf{\textsf{let}}^{\bm{*}}\ {#1} {#2}\ \textbf{\textsf{in}}\;{#3}\;}
\newcommand{\Letrec}[3]%
    {\textbf{\textsf{letrec}}\ {#1} {#2}\ \textbf{\textsf{in}}\;{#3}\;}

\newcommand{\ie}{\textit{i.e.}}

\newcommand{\QED}{{\Large $\square$}}

\newcommand{\size}[1]{|\,#1\,|}

\newcommand{\spacing}[2]{
  \renewcommand{\baselinestretch}{#2}
  \small\normalsize #1
  \setlength{\parskip}{0.1\baselineskip}
  \settowidth{\parindent}{xxxx}
  \setlength{\parindent}{#2\parindent}
  \setlength{\leftmargini}{\parindent}
  \setlength{\leftmarginii}{\parindent}
  \setlength{\leftmarginiii}{\parindent}
  \setlength{\footnotesep}{#2\footnotesep}
}

\begin{document}
\spacing{\normalsize}{.98}
\setcounter{page}{1}     
\setcounter{tocdepth}{1} 
\ifTR
  \pagenumbering{roman} 
\else
\fi

\title{Shortest path and maximum flow problems in planar flow networks with additive gains and losses}

\author{Saber Mirzaei %
           \thanks{Partially supported by NSF awards CNS-1135722 and ECCS-1239021.} \\
        Boston University \\
        \ifTR Boston, Massachusetts \\
        \href{mailto:kfoury@bu.edu}{kfoury{@}bu.edu}
        \else \fi
\and
       Assaf Kfoury %
          \footnotemark[1]\\
       Boston University  \\
        \ifTR Boston, Massachusetts \\
        \href{mailto:smirzaei@bu.edu}{smirzaei{@}bu.edu}
        \else \fi
}

   \date{\today}
\maketitle
  \ifTR
     \thispagestyle{empty} 
  \else
  \fi

\vspace{-.3in}
  \begin{abstract}
  In contrast to traditional flow networks, in additive flow networks,
to every edge $e$ is assigned a gain factor $g(e)$ which represents
the loss or gain of the flow while using edge $e$.  Hence, if a flow
$f(e)$ enters the edge $e$ and $f(e)$ is less than the designated
capacity of $e$, then $f(e) + g(e) \ge 0$ units of flow reach the end
point of $e$, provided $e$ is used, \ie, provided $f(e)\neq 0$.  In
this report we study the maximum flow problem in additive flow
networks, which we prove to be NP-hard even when the underlying graphs
of additive flow networks are planar.  We also investigate the
shortest path problem, when to every edge $e$ is assigned a cost value
for every unit flow entering edge $e$, which we show to be NP-hard in
the strong sense even when the additive flow networks are planar.

  \end{abstract}

\ifTR
    \newpage
    \tableofcontents
    \newpage
    \pagenumbering{arabic}
\else
    \vspace{-.2in}
\fi
\section{Introduction}
\label{sect:Introduction}
  In traditional flow network problems, such as the max-flow problem,
it is assumed that if $f(e)$ units of flow enter the edge $e=(v, u)$
at its tail $v$, exactly $f(e)$ units will reach its head $u$.
In practice this assumption in many flow models does not hold. For instance, in the well-known
\emph{generalized flow networks}, if $f(e)$ units of flow enter $v$, and a
gain factor $g(e)$ is assigned to $e$, then $g(e)\times f(e)$ units reach $u$.
Depending on the application, the gain factor can represent the
loss or gain due to evaporation, energy dissipation, interest,
leakage, toll or \emph{etc}.

The \emph{generalized maximum flow} problem
has been widely studied. Similar to the standard max-flow problem,
generalized max-flow problem can be formulated as a linear programming, and therefore
it can be polynomially solved using different approaches such as
modified simplex method, or
the ellipsoid method, or the interior-point methods.
Taking advantage of the structure of the problem,
different general purpose linear programming algorithms have been
tailored to speed up the calculation of max-flow in generalized flow networks~\cite{kamath1995improved,kapoor1996speeding,murray1992interior}.

The strong relationship between generalized max-flow problem and \emph{minimum cost flow} problem
was first recognized and established by Truemper in~\cite{truemper1977max}.
Exploiting this relationship and more importantly the discrete structure of the underlying graph, the generalized max-flow problem
can also be solved in polynomial time by \emph{combinatorial methods}~\cite{Vegh2014, goldfarb2002combinatorial, onaga1966dynamic,goldberg1991combinatorial,goldfarb2002polynomial,tardos1998simple}.

In contrast to the well-studied generalized flow networks, recently in~\cite{brandenburg2011shortest}, the authors
introduced and investigated flow networks where an additive fixed gain factor is
assigned to every edge \emph{if used}. Flow networks with additive gains and losses (\emph{additive flow networks} for short) have several applications in practice.
In communication networks, a fixed-size load is added to every package being sent out by routing nodes
in the network. In transportation of goods or commodities, a fixed amount may be lost
in the transportation process or a flat-rate amount of other commodities or cost may be added to
the commodity passing every toll station. In financial systems there are fixed costs or losses for
every transaction.

The max-flow problem in additive networks can be views as the problem of finding a
feasible flow which either maximizes the amount of flows departing the source vertices
or maximizes the amount of flow reaching the sink vertices (respectively called \emph{maximum in-flow} and \emph{maximum out-flow} problems).
Similarly, assuming that a \emph{unit flow} is departing a source vertex, the shortest path problem is the problem of finding
a feasible flow along a path from the source vertex to the sink vertex with minimum accumulated cost.

As explained in~\cite{brandenburg2011shortest}, flow networks with additive gains and losses,
are different in many aspects from standard and generalized flow networks due to
some properties such as \emph{flow discontinuity}, lack of
max-flow/min-cut \emph{duality}, and unsuitability of \emph{augmented path} methods.
Similarly, some basic properties of the shortest path in standard flow network do not hold in the additive case.
For instance, it is not anymore the case that the sub-path, the prefix or the suffix of a shortest path must themselves be shortest.
For more details on the properties of additive flow networks, we refer the reader to~\cite{brandenburg2011shortest}.
These differences make both the max-flow problem and the shortest path problem hard to solve for additive flow networks.
Precisely speaking, the authors of the same paper show that the shortest path problem and the maximum in/out-flow problems are NP-hard for general graphs.
In this paper we extend their results to the case where the underlying graph of the flow network is planar.

\vspace{-12 pt}
\paragraph{Organization of the Report.} In Section~\ref{sect:definitions} we give precise formal definitions of several notions
regarding the additive flow networks and the corresponding problems.
Section~\ref{sect:shortestPath} concerns the shortest path problem in the additive flow networks; in this section we show that this problem is NP-hard in the strong sense when the underlying
graph of the network is planar and Section~\ref{sect:maxFlow} extends the NP-hardness of maximum in/out-flow problems
to planar additive flow networks. Finally, Section~\ref{sect:future} is a brief preview of future work.

\section{Definitions and preliminaries}
\label{sect:definitions}
  A flow network with additive losses and gains (\emph{additive flow network} for short) is defined by a tuple $N=(V, E, S, T, u, c, g)$, where $G=(V, E)$ is the directed underlying graph with
$n = |V|$ vertices and $m = |E|$ edges. The two sets $S, T \subset V$ are respectively the designated set of source and sink vertices.
$u: E \rightarrow \mathbb{R}^{+}$ is the edge capacity function, while $c: E \rightarrow \mathbb{R}$ is the cost function and $g:E \rightarrow \mathbb{R}$ assigns a gain or loss value to every edge \emph{if used}.
Precisely speaking, the cost $c(e)$ per units of flow and the gain factor $g(e)$ are applied, only if a positive flow enters the edge $e$.
\begin{definition}{Incoming and outgoing edges} Consider the directed graph $G=(V, E)$ and $v \in V$. Two functions $in:V \rightarrow 2^{E}$ and $out:V \rightarrow 2^{E}$ respectively represent the set of incoming edges and the set of outgoing edges for every vertex. Formally, for every $v \in V$:
\begin{align*}
in(v) = \Set{e | e = (u,v) \in E}, \\
out(v) = \Set{e | e = (v, u) \in E}.
\end{align*}
Similarly the \emph{indegree} and \emph{outdegree} functions for every vertex $v \in V$ are respectively defined as $deg^{+}(v) = \size{in(v)}$ and $deg^{-}(v) = \size{out(v)}$.
\end{definition}
As in traditional flow networks, it is assumed that for every source vertex $s \in S$, $deg^{+}(s) = 0$ and for every sink vertex $t \in T$, $ deg^{-}(t) = 0$.
Flow $f: E \rightarrow \mathbb{R}^{+}$ in a network $N$ is feasible if it satisfies edge capacity $0 \le f(e) \le u(e)$ for every edge $e \in E$
and flow conservation constraint at every vertex in $V$. Formally speaking, $f$ satisfies the flow conservation constraints if for every $v \in V - (S \uplus T)$:
\begin{align*}
\sum_{e \in in(v) , f(e) > 0} \max(0, f(e)+g(e)) =  \sum_{e  \in out(v)} f(e).
\end{align*}
Given edge $e = (v,u)$, if flow $f(e)$ exits vertex $v$, then $\max(0, f(e)+g(e))$ units reach $u$. Therefore edge $e =(u, v) \in E$ is lossy if the entering edge $f(e)$ is positive and $g(e) < 0$. Edge $e$ consumes the entering flow if $f(e)+g(e) < 0$, in which case no flow reaches the vertex $v$.
\begin{definition}{Out-flow and in-flow}
Given flow $f$ for network $N$, the \emph{out-flow} is the summation of the amount of flow exiting source vertices, \emph{i.e.}, $f_{out} = \sum_{s \in S, e \in out(s)} f(e)$. Similarly, \emph{in-flow} is the summation of flow values entering all sink vertices, namely $f_{in} = \sum_{t \in T, e \in in(t) , f(e) > 0 , } max(0, f(e) + g(e))$.
\end{definition}
In additive flow networks, the max-flow problem can be studied from the producers' (source vertices) point of view or from the consumers' (sink vertices) point of view. Hence, the maximum out-flow problem is the problem of finding a feasible flow $f$ maximizing $f_{out}$, and the maximum in-flow problem is defined similarly. Regarding these two problems, while trying to maximize the amount of outgoing/incoming flows, we are not concerned with the cost of flow.

On the other hand, in additive flow networks, the shortest path problem can be generalized in several ways. For instance, given a producer vertex $s \in S$ (similarly consumer vertex $t \in T$), the problem can be defined as finding a consumer vertex $t \in T$ (producer vertex $s \in S$) with the shortest distance among the others, with respect to the cost of a unit flow departing $s$ towards $t$. A more generalized variation of the problem can be defined, where neither of the source or destination vertices are fixed. Therefore, the generalized shortest path is the problem of finding a source vertex $s \in S$, a destination vertex $t \in T$, and a path $\Pi$ from $s$ to $t$ with minimum cost, if a unit of flow departs $s$.
Without loss of generality, in the rest of this report, it is assumed that the sets of source and sink vertices each has one member (namely $|S| = |T| = 1$). Hence, every negative result shown for this special case is immediately applicable to the generalized version of the shortest path problem.
Section~\ref{sect:shortestPath} concerns the shortest path problem and the related definitions with more details.

\begin{definition}{Cost of flow}
Given a flow function $f$ in network $N$, the \emph{cost of flow} on every edge $e$ is $f(e) \times c(e)$.
Similarly the \emph{accumulated cost} of $f$ is the summation of cost of flow entering all edges, \emph{i.e.}, $cost(f) = \sum_{e \in E} f(e) \times c(e)$.
\end{definition}
\begin{example}
\label{ex:additiveFlow}
In Figure~\ref{fig:additiveFlow} an additive flow network and a feasible flow from vertex $s$ to vertex $t$ are depicted.
The cost of unit flow for every edge is $1$ and the capacity of every edge is $B+1 $ for $B > 1$.
For a compact illustration of figures in this report, we adopt the following conventions:
\begin{itemize}
  \item If we label an edge $e$ with $g = r$ or $c = r$ or $u = r$, for some $r \in \mathbb{R}$, then we mean that $g(e) = r$ or
  $c(e) = r$ or $u(e)=r$, respectively.
  \item If we omit such a label on an edge $e$, then we mean that the value of the corresponding function for $e$ is the default value (as stated in the description of that figure).
\end{itemize}
For instance in Figure~\ref{fig:additiveFlow}, $g((s,v_1)) = B + 1$ for edge $(s,v_1)$, is represented by $g = B+1$ by that edge. The missing gain values are the default value $0$.
In this example, the initial unit flow leaves
vertex $s$ while $B+2$ units reach $v_1$, due to gain value $B+1$ assigned to edge $(s,v_1)$. The flow entering edge $(v_1,v_5)$ is fully absorbed by the gain factor $-B$ assigned to it. Hence, the accumulated cost of the flow along path\footnote{A simple path can be interchangeably represented both by the set of vertices or by the set of edges taking part in it.} $\Pi=(s,v_1,v_2,v_3,t)$
is $B+4$ while the accumulated cost of the flow $f$ is $B+5$.
\begin{figure}[H]
    \begin{center}
        \includegraphics[scale=0.65]{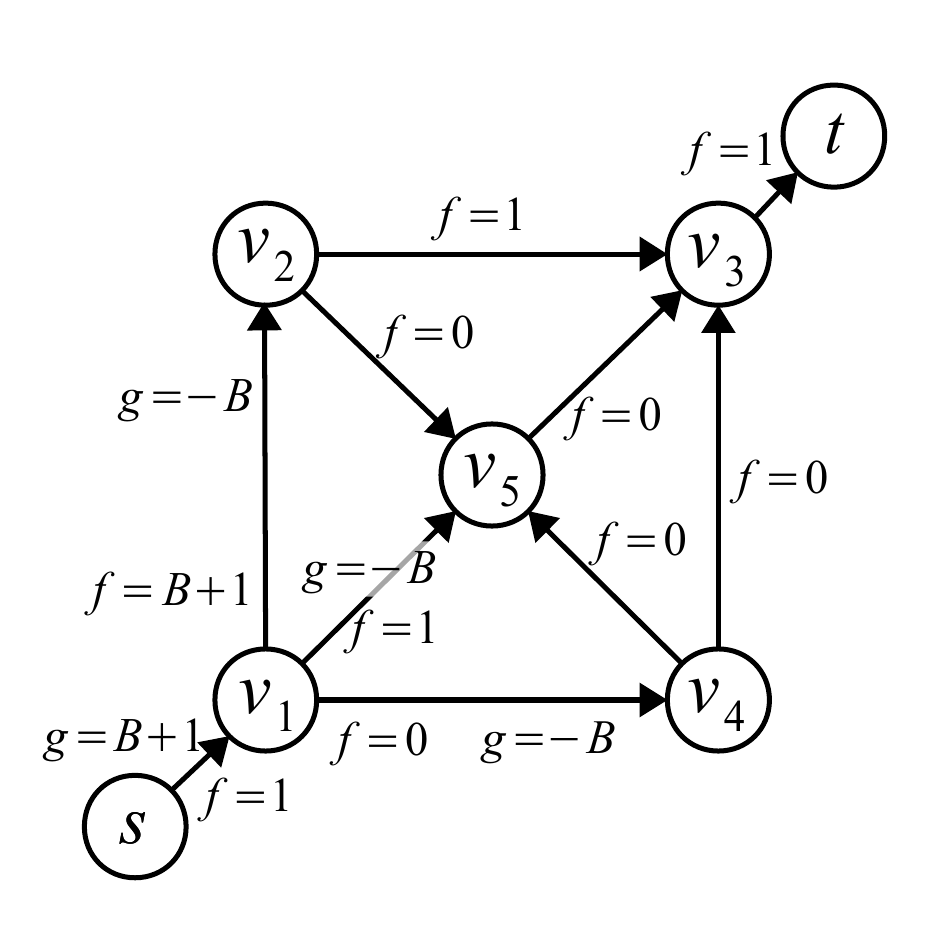}
    \end{center}
    \vspace{-0.4in}
    \caption{An example of an additive flow network $N$ and the flow function $f$ assigning feasible flow values to every edge of $N$.
    Missing gain values are $0$, the cost function $c$ is $1$ on every edge, and the capacity function $u$ is a ''large number'' on every edge (for example $B+1$).}
    \label{fig:additiveFlow}
    \vspace{-0.1in}
\end{figure}
\end{example}

\section{Shortest path problem in additive flow network}
\label{sect:shortestPath}
  Let $N$ be an additive flow network and let $\Pi=\Parens{(s,v_1), (v_1,v_2),\ldots,(v_{k-1},t)}$ be a (simple) path in $N$ from vertex $s$ to $t$.
Let $e_i = (v_{i-1}, v_i)$ for $2 \le i \le k-1$ while $e_1 = (s, v_1)$ and $e_k = (v_{k-1},t)$.
Given the flow $f$ along path $\Pi$ with the initial flow $f_1$ entering $e_1$ (\emph{seed flow} for short), the accumulated flow entering edge $e_i$ for every $1 \le i \le k$, is represented by:
\begin{align*}
\gamma(\Pi,f_1, i) = f_1 + \sum_{j < i} g(e_1).
\end{align*}
\begin{definition}{Feasible and dead-end flows along a path}
Flow  $f$ with the seed flow $f_1$, is \emph{feasible} along path $\Pi = (e_1,\ldots, e_k)$ if the flow on every edge is positive and a positive amount of flow reaches the destination $t$. \emph{i.e.} for $1 \le i \le k+1$:
\begin{align*}
\gamma(\Pi,f_1, i) > 0.
\end{align*}
Flow $f$ is \emph{infeasible} or \emph{dead-end} if an edge $e_i$ for $1 \le i \le k$
absorbs all the entering flow, due to the loss value $g(e_i)$ assigned to that edge. Hence, no flow reaches the destination vertex.
\end{definition}
For every vertex $t$ reachable from $s$ via some path $\Pi$, there exists a threshold value $T$ such that a flow along $\Pi$ is feasible only if its seed flow  $f_1 > T$.
Finding the reachability threshold (if exists) for every pair of source and sink vertices is a polynomial task, as formally stated in Lemma~\ref{lemma:threshold}.
\begin{lemma}
\label{lemma:threshold}
Consider a flow network $N = (V,E, \Set{s}, \Set{t}, u, c, g)$. There exists a polynomial time algorithm that decides the reachability of $t$ from $s$ in network $N$ and finds the reachability threshold in $\bigO{nm}$, if such threshold exists.
\end{lemma}
\begin{proof}
In order to find the threshold value $T$ in $N$ such that at least one path from source $s$ to destination $t$ is feasible (\emph{i.e.},
$t$ is reachable from $s$ for seed flow $f_1 > T$), one can use a variation of well-known shortest path algorithms (such as Bellman-Ford algorithm) in the reversed graph of $N$. The reversed graph is constructed from the underlying graph of $N$ by reversing the direction of every edge $e$ and assigning $-g(e)$ as the weight of the reversed edge. Hence, the reachability threshold problem reduces to a modified variation
of the shortest path problem from $t$ to $s$. This problem is a modified variation of the traditional shortest path problem with negative costs, in the sense that the distance of no two vertices can be negative.
Hence, in the modified variation of the Bellman-Ford algorithm in the process of updating the distance matrix for every vertex $v$ from $s$,
the distance of $v$ from $s$ is set to $\max \Set{0,d}$, where $d$ is the newly updated distance of $v$ from $s$.
Note that due to flow conservation constraints at every vertex, there is no feasible flow with positive gain cycles involved in it. Also in this modified variation of Bellman-Ford algorithm the distance of no two vertices can be negative.
Assuming that there is no positive gain cycle in $N$, means that there is no negative cycle in the reversed graph, which in turn implies that the modified variation of Bellman-Ford algorithm always returns the threshold, if $t$ is reachable from $s$. The correctness proof of this approach is straightforward and similar to the proof of the correctness of the standard Bellman-Ford algorithm for the shortest path problem.
The time complexity of this algorithm is the same as the standard Bellman-Ford algorithm, which has $\bigO{|V|\times|E|}$ worst case time bound.
\end{proof}
The accumulated cost of a flow along the path $\Pi = (e_1,\ldots, e_k)$ (feasible or not), is the summation of the cost of the flows entering every edge times the value of cost function assigned to that edge, namely:
\begin{align*}
\sum_{1 \le i \le k} c(e_i) \gamma(\Pi,f_1, i).
\end{align*}
In~\cite{brandenburg2011shortest} regarding the shortest path problem, the authors simplify the problem by assuming that
the seed value $f_1 = 1$.
On the other hand, given the source and destination vertices $s$ and $t$ in flow network $N$, it may be the case that for every path $\Pi$ from $s$ to $t$, no feasible flow along $\Pi$ with the seed value $f_1 = 1$ exists. Accordingly the definition of shortest path problem in~\cite{brandenburg2011shortest} can be generalized as in following.
\begin{definition} {Shortest path in additive flow networks}
\label{def:shortestPath}
The \emph{shortest path} problem in additive flow network $N$ for a given
pair of source and sink vertices $\Set{s,t}$, is the problem of finding a min-cost feasible flow along some path $\Pi$ from $s$ to $t$, when the
seed flow $f_1 = \min\Set{1, T}$. Where $T$ is the reachability threshold for the source/destination pair $\Set{s,t}$.
\end{definition}
In contrast to the shortest path problem in traditional flow networks, this problem is hard in the case of additive flow networks.
From Theorem 1 in~\cite{brandenburg2011shortest} it can be inferred that when the underlying graph of the flow network is planar, this problem is weakly NP-hard. In other words,
the problem may be polynomially solvable if the cost and capacity values assigned to the edges are bounded from above by some polynomial function in the
size of the graph. In the same paper it is shown that this result holds if the cost and gain values are all \emph{nonnegative integers}, even for the case where the underlying graph is not necessary planar.

In this section we show that the shortest path problem is NP-hard \emph{in the strong sense} when the underlying graph of the additive flow network is planar (\emph{planar additive flow network} for short).
We show this result using a polynomial reduction from a problem called Path Avoiding Forbidden Transitions (PAFT for short)~\cite{kante2015finding}.
PAFT is a special case of the problem of finding a path from source vertex $s$ to destination vertex $t$ while avoiding a set of forbidden paths
(initially introduced in~\cite{villeneuve2005shortest}).
Before presenting the main result of this section (as stated in Theorem~\ref{thr:shortestPath}), we briefly define PAFT and some results on this problem (for more details, the reader is referred to~\cite{kante2015finding}).

Given undirected multi-graph $G=(V, E)$, a \emph{transition} in $G$ is an unordered set of two distinct edges of $E$ which are incident to the same vertex of $V$.
If $\mathcal{T}$ denotes the set of all possible transitions in the graph $G$, the set $\mathcal{F}$ of \emph{forbidden transitions} is a subset of $\mathcal{T}$.
Then $\mathcal{A} = \mathcal{T} - \mathcal{F}$ denotes the set of \emph{allowed transitions}.
A simple path $\Pi=(e_1, e_2,\ldots,e_k)$, where $ e_1 = \Set{s,v_1}$ and $e_k = \Set{v_{k-1}, t}$, is $\mathcal{F}$-\emph{valid} if for every $1 \le i < k$,
$\Set{e_i,e_{i+1}} \notin \mathcal{F}$; namely, no transition in $\Pi$ is forbidden. A vertex $v$ is \emph{involved} in a forbidden transition
$\Set{e, e^{\prime}}$
if two edges $e$ and $e^{\prime}$ share $v$ and $\Set{e, e^{\prime}} \in \mathcal{F}$.
\begin{definition}{PAFT}
\label{def:PAFT}
Consider a multi-graph $G = (V, E)$, a set of forbidden transitions $\mathcal{F}$ and designated source and destination vertices $s, t \in V$.
PAFT is the problem of find an $\mathcal{F}$-valid path from $s$ to $t$, if exists.
\end{definition}
\begin{lemma}
\label{lemma:PAFT}
PAFT is NP-complete for planar graphs where the degree of every vertex $v \in V -\Set{s,t}$ is $3$ or $4$, where $s$ and $t$ are the source and destination vertices, respectively.
\end{lemma}
\begin{proof}
In~\cite{kante2015finding}, the authors show that PAFT is NP-complete in planar graphs where the degree of every vertex is at most $4$.
Consider graph $G = (V, E)$ with maximum vertex degree $4$ and source and destination vertices $s$ and $t$, and let $\mathcal{F}$  be the set of forbidden transitions.
Graph $G^{\prime} = (V^{\prime}, E^{\prime})$ and the set of forbidden transitions $\mathcal{F}^{\prime}$ are constructed as explained in what follows.

\noindent
Initially $V^{\prime} = V$, $E^{\prime} = E$ and $\mathcal{T}^{\prime} = \mathcal{T}$. Until there is no vertex of degree $1$ or $2$, for every vertex $v \in V^{\prime} - \Set{s,t}$:
\begin{enumerate}
  \item \textbf{If $deg(v) = 1$:} (i) $V^{\prime} = V^{\prime} -{v}$, (ii) $E^{\prime} = E^{\prime} - e$, where $e$ is the edge incident to $v$, and (iii) $\mathcal{T}^{\prime} = \mathcal{T}^{\prime} - \tau$, for every forbidden transition $\tau$ that $v$ is involved in.

  \item \textbf{If $deg(v) = 2$ and $\tau = \Set{e_1, e_2} \in \mathcal{F}$, where $e_1$ and $e_2$ share $v$:} (i) $V^{\prime} = V^{\prime} -{v}$, (ii) $E^{\prime} = E^{\prime} - \Set{e_1, e_2}$, and (iii) $\mathcal{T}^{\prime} = \mathcal{T}^{\prime} - \tau$

  \item\textbf{If $deg(v) = 2$ and $\tau = \Set{e_1, e_2} \notin \mathcal{F}$, where $e_1$ and $e_2$ share $v$:} (i) $v$ is smoothed out by removing $v$ and the two incident edge $e_1 =\Set{w,v}$ and $e_2 = \Set{v, u}$ and introducing a new edge $e^{\prime} = \Set{w, u}$.
      (ii) For $i \in \Set{1,2}$ if $\Set{e_i, e} \in \mathcal{F}$ for some $e \in E$; $F^{\prime} = F^{\prime}\uplus \Set{\Set{e^{\prime}, e}} - \Set{e_i, e}$.

\end{enumerate}
It is straightforward and left to reader to verify that there is $\mathcal{F}$-valid path from $s$ to $t$ in $G$ iff there is a $\mathcal{F^{\prime}}$-valid path from $s$ to $t$ in $G^{\prime}$.
Accordingly, the NP-hardness of PAFT for planar graphs with maximum degree $4$ results in the NP-hardness of PAFT for the class of planar graphs where the degree of every vertex is $3$ or $4$. In~\cite{kante2015finding}, using a similar approach, this result is extended to grid graphs.
\end{proof}

This lemma helps us to draw the main result of this section (as stated at the end of this section in Theorem~\ref{thr:shortestPath}).
Before that, in an intermediate step, Procedure~\ref{proc:PAFT-to-Shortest} represents an approach to transforming an instance of PAFT into an additive flow network. Without loss of generality, we assume that the source and the sink vertices  are not involved in any
forbidden transition\footnote{If the source vertex (or sink vertex) is involved in any forbidden transition, a new source vertex (or sink vertex) is introduced and is connected to the old one.}.
\begin{procedure}{PAFT's instance into additive flow network}
\label{proc:PAFT-to-Shortest}
Consider a planar undirected multi-graph\footnote{Assume that the planar embedding is given.} $G = (V, E)$, a set of forbidden transitions $\mathcal{F}$, and source and destination vertices $s$ and $t$
as an instance of PAFT, where for every vertex $v \in V$, $deg(v) \in \Set{3,4}$.
We transform such instance of PAFT into an additive flow network $N$ in several steps:
\begin{enumerate}
  \item Every undirected edge $e = \Set{v,u}$ is replaced by a pair of parallel incoming/outgoing directed edges $e=(v,u)$ and $e^{\prime}=(u,v)$.

  \item Two new vertices $s^{\prime}$ and $t^{\prime}$ are introduced. $s^{\prime}$  is connected to $s$ via edge $e_s = (s^{\prime}, s)$ and $t$ is connected
  to $t^{\prime}$ via $e_t = (t, t^{\prime})$. $e_t$ has $0$ gain while $e_s$ has gain value $B$ assigned to it for some $B > 1$.
  To both edges $e_s$ and $e_t$ is assigned cost value $0$.

  \item For every vertex $v$, not involved in any forbidden transition, every outgoing edge $(v,v^{\prime})$ is subdivided into two edges
  $(v, w)$ and $(w, v^{\prime})$ by introducing a new vertex $w$. Two edges $(v, w)$ and $(w, v^{\prime})$
  respectively have gain value $-B$ and $+B$ and cost $c = + B$ and $c = - B$. Figure~\ref{fig:non-forbidden} represents the replacement gadget for a vertex $v$ with degree $4$.

  The gain value $-B$ assigned to every outgoing edge $(v, w)$ guarantees that if $B+1$ units of flow enter the gadget of $v$,
  the exiting flow reaches the gadget of at most one of the neighbors of $v$ in $G$. The gain value $+B$ assigned to the edge $(w, v^{\prime})$
  (connected to outgoing edge $(v, w)$) compensates for the lost flow that enters $(v, w)$, if some flow reach $w$. Hence, $B+1$ units of flow entering the gadget of $v$ reach the gadget of exactly one of the neighbors of $v$ in $G$ \emph{with no loss}, if only one of the outgoing edges is chosen.

  Moreover based on the same reasoning, the cost values in this gadget assure us that the $B+1$ units of flow entering this gadget reaches the gadget of the neighboring vertex \emph{with no cost}, if only one of the outgoing edges is chosen.

  In summary, if $B+1$ units of flow enter such gadget of a vertex $v$, in order to have some flow reaching the neighboring gadget, one of the
  outgoing edges must have $B+1 - x$ units entering flow for $0 \le x < 1$. In this case $B+1-x$ units reach the neighboring gadget and the $x$ units of flow is consumed with total cost $xB$.
\begin{figure}[H]
    \begin{center}
        \includegraphics[scale=0.5]{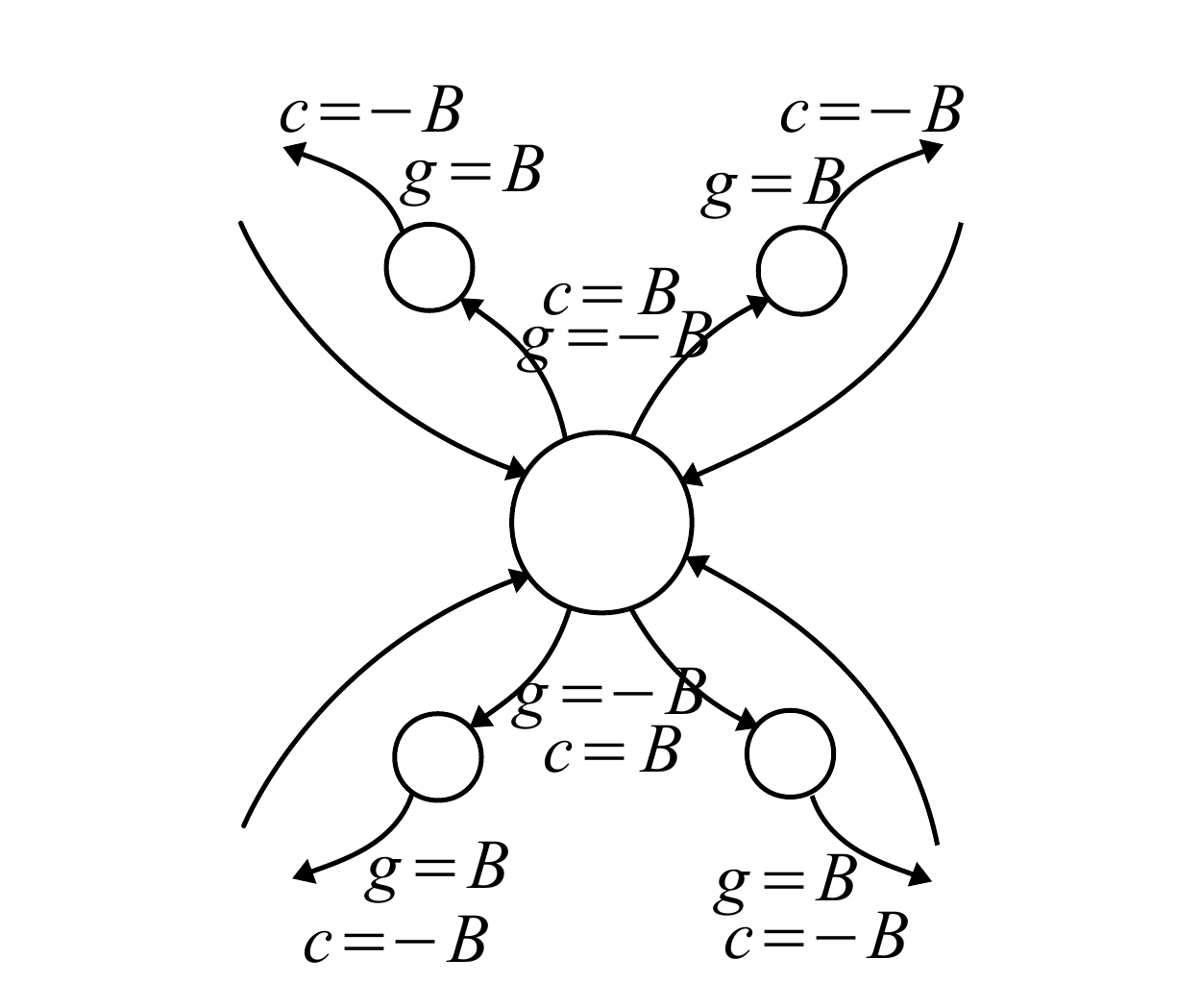}
    \end{center}
    \vspace{-0.3in}
    \caption{The gadget that replaces a vertex of degree $4$ which is not involved in any forbidden transition.}
    \label{fig:non-forbidden}
    \vspace{-0.15in}
\end{figure}
  \item Consider vertex $v$ involved in some forbidden transitions with degree $d = deg(v) \le 4$ incident to $d$ pairs of incoming/outgoing parallel edges $\Set{e_1, e_1^{\prime}} ,\ldots,\Set{e_{d}, e_{d}^{\prime}}$. Originally in graph $G$ vertex $v$ is incident to $e_1, \ldots, e_{d}$. The planar embedding of $v$ in $G$ is represented in Figure~\ref{fig:forbidden} on the left, when $d = 4$.
      \begin{enumerate}
        \item \textbf{If $d = 3$}:
        vertex $v$ is replaced by $3$ new vertices $v_1,v_2, v_{3}$ where $v_i$ is incident to the pair $\Set{e_i, e_i^{\prime}}$ for $1 \le i \le 3$, based on the planar embedding of edges $e_1, e_2, e_{3}$ in $G$.
        For $1 \le i, j \le 3$, two vertices $v_i$ and $v_j$ are directly connected with a pair of parallel incoming/outgoing edges, if
        $\Set{e_i,e_j} \notin \mathcal{F}$. The cost of every edge is $0$ and the gain factor for every introduced edge $(v_i, v_j)$ is $-B$ for $1 \le i, j \le 3$. Finally, every outgoing edge connecting $v_i$  (for $1 \le i \le 3$) to another gadget (corresponding to one of the neighbors of $v$ in $G$) has cost $0$ and gain factor $+B$.
        \item \textbf{If $d = 4$ and $\Set{\Set{e_1,e_4}, \Set{e_2,e_3}} \nsubseteq \mathcal{A} $}: A gadget of four vertices replaces $v$, following the same procedure as explained in the previous case.
        Figure~\ref{fig:forbidden} depicts an example of transforming a vertex of degree $4$ involved in some forbidden transitions.
\begin{figure}[h]
    \begin{center}
        \includegraphics[scale=0.5]{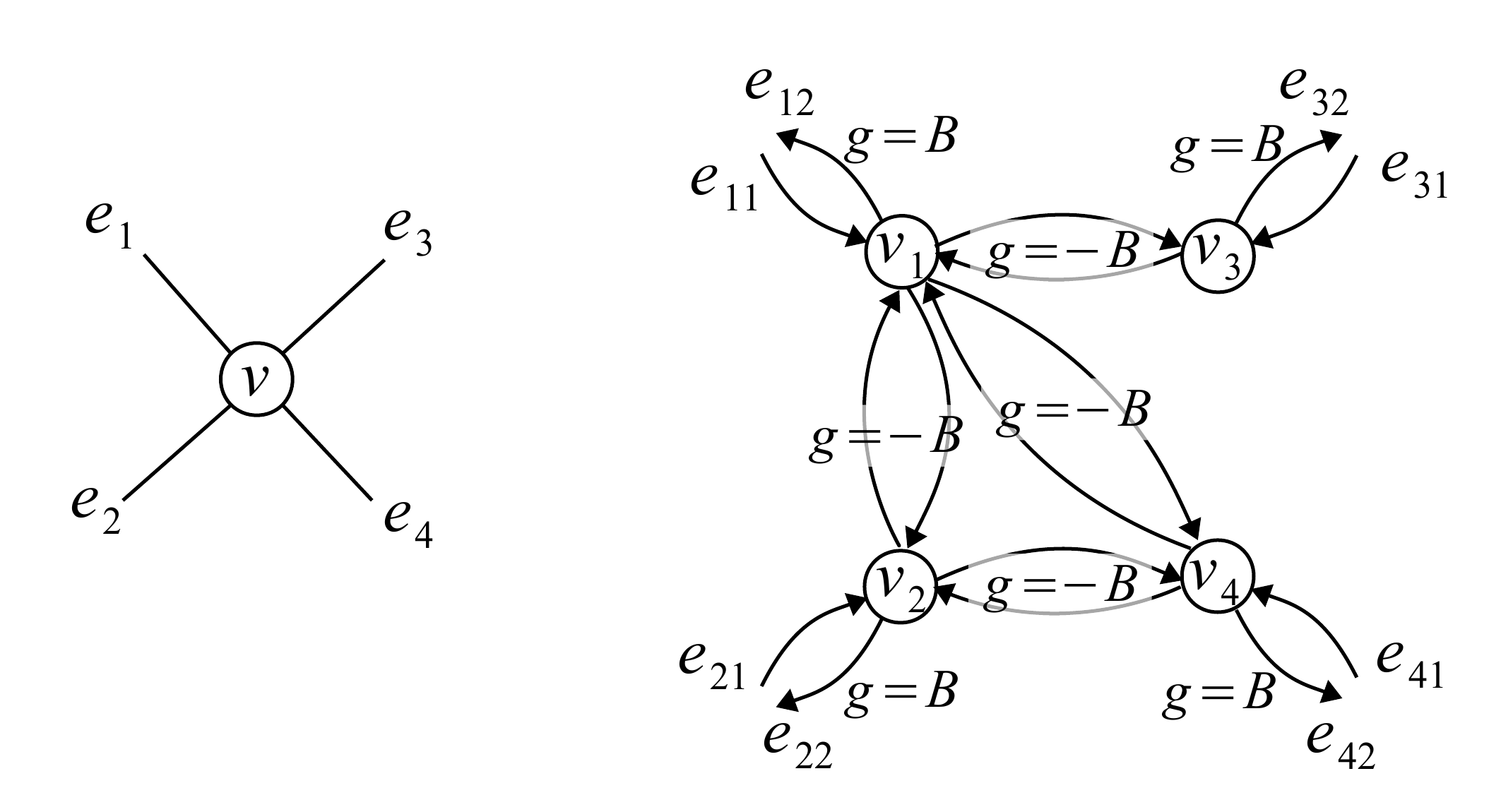}
    \end{center}
    \vspace{-0.3in}
    \caption{On the left is the planar embedding of a vertex $v \in V$ incident to $4$ edges where $\Set{e_2, e_3}, \Set{e_3, e_4} \in \mathcal{F}$.
    The right image shows the gadget replacing vertex $v$. Every edge connecting two vertices $v_i$ and $v_j$ for $1 \le i, j \le 4$
    has cost value $0$ and gain value $-B$ assigned to it.
    Vertices $v_2$ and $v_3$ are not directly connected since their corresponding
    edges $e_2$ and $e_3$ are involved in a forbidden transition. Hence, a flow from $v_2$ to $v_3$ (and vice versa) loses $-2B$ units.
    Same situation holds for $v_3$ and $v_4$.
    }
    \label{fig:forbidden}
    \vspace{-0.2in}
\end{figure}
        \item \textbf{If $d = 4$ and $\Set{\Set{e_1,e_4}, \Set{e_2,e_3}} \subseteq \mathcal{A} $:}
        Using the same approach as in the previous case spoils the planarity of the resulting network $N$.
        To solve this problem, vertex $v$ is replaced by a gadget of $5$ vertices.
        In addition to $4$ vertices $v_1, \ldots, v_4$, where for $1 \le i \le 4$, $v_i$ is connected to the pair of incoming/outgoing edges $\Set{e_i, e_i^{\prime}}$,
        a central vertex $w$ is introduced as well.
        As depicted in Figure~\ref{fig:forbidden-bad}, in the replacing gadget,
        vertices $v_1$ and $v_4$ (also vertices $v_2$ and $v_3$) are connected via the central vertex $w$.
        Any other two vertices $v_i$ and $v_j$
        for $1 \le i, j \le 4$ are directly connected via a pair of parallel incoming/outgoing edges, if
        $\Set{e_i,e_j} \notin \mathcal{F}$ (with gain and cost values $-B$ and $0$ respectively).
        The cost and gain factors for every edge can be found by that edge and the missing values are the default value $0$.

        Assume $B+1$ units of flow reach any of the four vertices $v_1, \ldots, v_4$. The gain and cost values assigned to the edges incident to $w$ guarantee that a flow can go through the central vertex $w$ \emph{with no cost} and reach the destination \emph{with $-B$ units loss}, only if it is from $v_1$ to $v_4$ and vice versa or if it is from $v_2$ to $v_3$ and vice versa. Any other flow, with the initial value $B+1$ units, that uses $w$
        is either costly (costs $2B(B+1)$) or gets fully consumed (\emph{i.e.}, does not reach the destination).
      \end{enumerate}

  \item Every edge has capacity $B+1$.
\end{enumerate}
\begin{figure}[h]
    \begin{center}
        \includegraphics[scale=0.4]{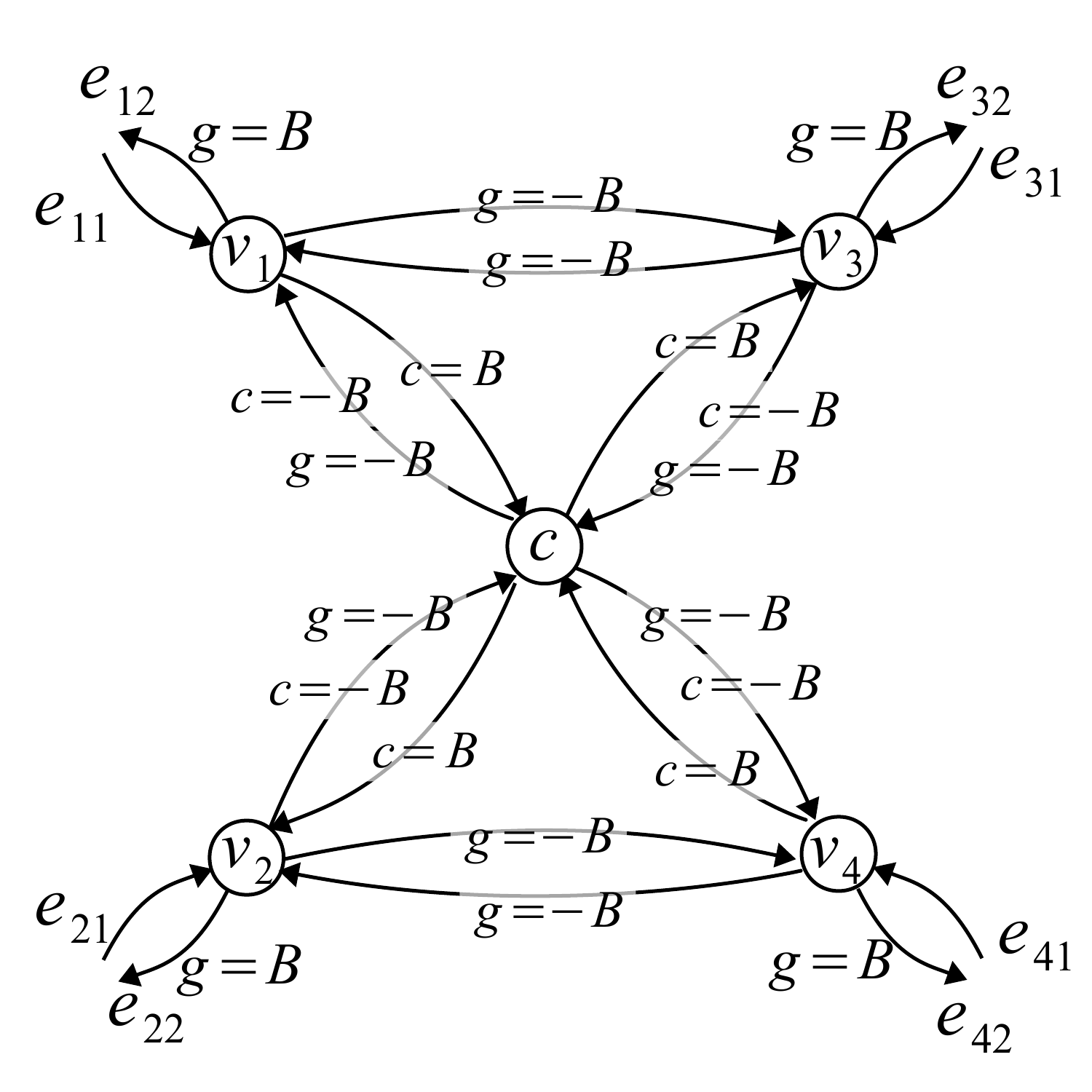}
    \end{center}
    \vspace{-0.3in}
    \caption{
    The gadget replacing vertex $v$ where $deg(v) = 4$ and $\Set{e_1,e_4},\Set{e_2,e_3} \in \mathcal{A}$. Edges $e_1, \ldots, e_4$ are the edges connected to $v$ as represented in a planar embedding of $G$ in Figure~\ref{fig:forbidden} (on the left). This gadget replaces $v$ where
    $\Set{e_1, e_2}, \Set{e_3, e_4} \in \mathcal{F}$.
    }
    \label{fig:forbidden-bad}
    \vspace{-0.2in}
\end{figure}
Note that if the undirected multi-graph $G$ is planar, then so is the underlying graph of the flow network $N$
as a result of the preceding transformation. Also since the capacity of every edge is $B+1$, it is guaranteed
that maximum amount of flow that enters a gadget is $B+1$ units and no more than $B+1$ units of flow departs any gadget.
\end{procedure}
\begin{example}Figure~\ref{fig:undirected-forbidden} denotes a graph $G$ and its designated source and destination vertices.
The set of forbidden transitions is $\mathcal{F}=\Set{\Set{e_2,e_3}, \Set{e_3, e_4}}$. Figure~\ref{fig:undirected-forbidden-after} shows
the additive flow network $N$ constructed based on $G$ and the set of forbidden transitions. The set of four vertices shown in the dashed
circle represents the gadget replacing vertex $v$ which is the only vertex involved in the forbidden transitions. The missing gains of every edge
in this gadget is $-B$.
\begin{figure}[H]
\hspace{35 pt}
        \begin{subfigure}[b]{0.3\textwidth}
                    \includegraphics[scale=.5]
                    {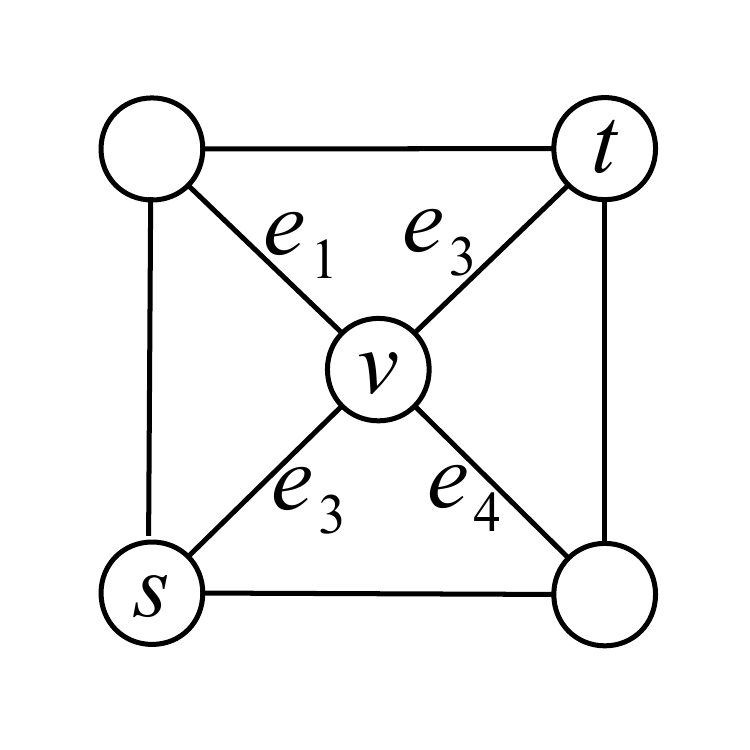}
                    \vspace{-0.2in}
                \caption{Graph $G$, input for PAFT.}
                \label{fig:undirected-forbidden}
        \end{subfigure}
        \qquad
        \begin{subfigure}[b]{0.3\textwidth}
                    \includegraphics[scale=.5]
                    {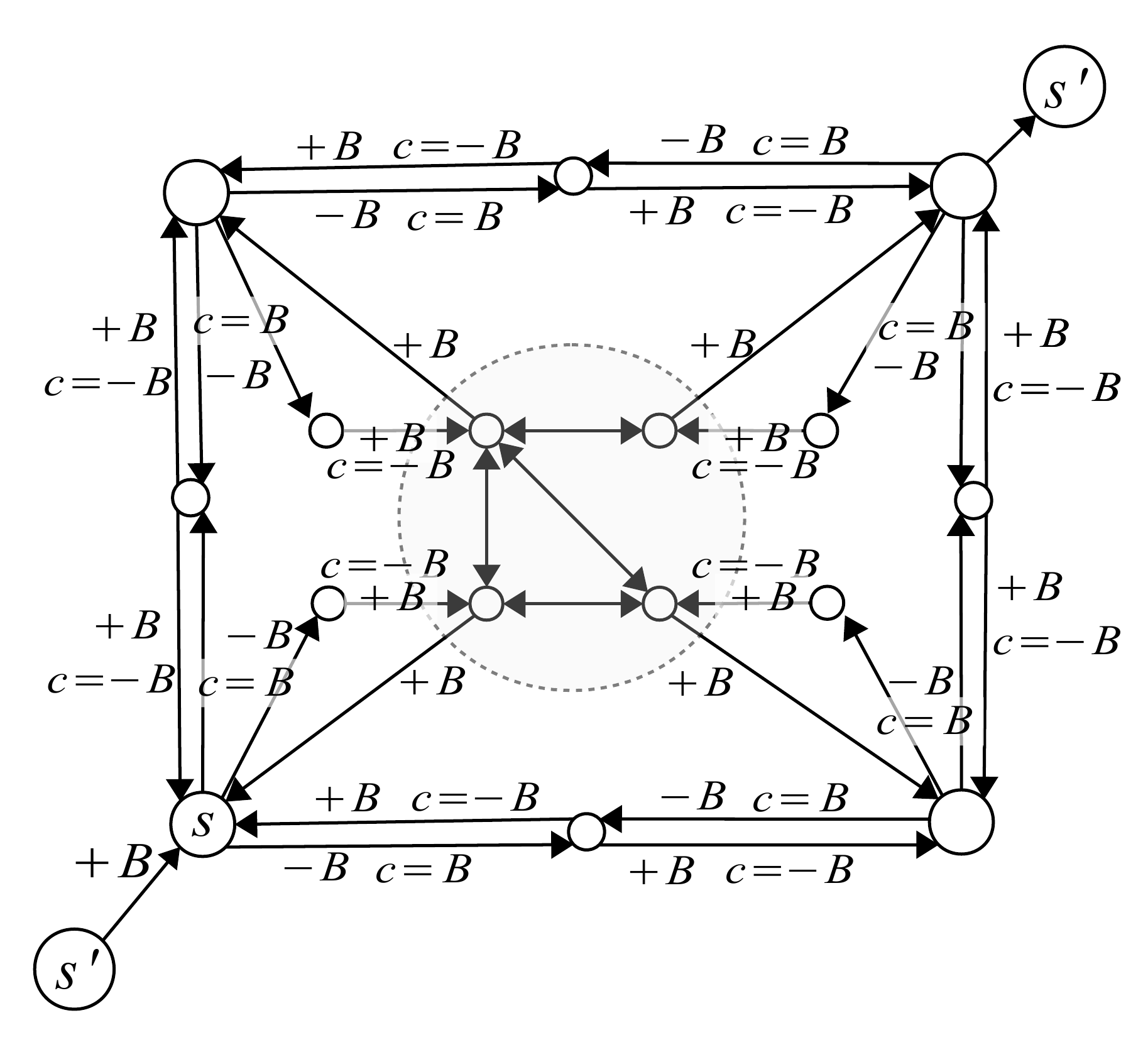}
                    \vspace{-0.2in}
                \caption{Additive flow network $N$.}
                \label{fig:undirected-forbidden-after}
        \end{subfigure}
        \caption{Additive flow network $N$ in~\ref{fig:undirected-forbidden-after}, constructed based on the graph $G$ in~\ref{fig:undirected-forbidden} and the set of forbidden transitions $\mathcal{F}=\Set{\Set{e_2,e_3}, \Set{e_3, e_4}}$.
        The set of four vertices shown in the dashed circle represent the gadget replacing middle vertex $v$.
        The missing gain values for every edge in this gadget are $-B$.
        The gain of every other edge is shown by $+B$ or $-B$ by that edge. The default values for missing cost and capacity functions are respectively $0$ and $B+1$.
        }
        \label{fig:undirected-forbidden-before-after}
        \vspace{-0.15in}
\end{figure}
\end{example}
\begin{lemma}
\label{lemma:PAFT-to-network}
Consider an undirected multi-graph $G$, a set of forbidden transitions $\mathcal{F}$, and source and destination
vertices $s$ and $t$ as an instance of PAFT, where the degree of every vertex $v \in V$ is $3$ or $4$. Let $N$ be the
additive flow network constructed from $G$ and the forbidden transitions set $\mathcal{F}$ using Procedure~\ref{proc:PAFT-to-Shortest}.
There is a (simple) $\mathcal{F}$-valid path from $s$ to $t$ iff  there exists a feasible flow in $N$ along a (simple) path from $s^{\prime}$ to $t^{\prime}$ with no cost, when seed flow $f_1 = 1$.
\end{lemma}
\begin{proof}
\textbf{($\Rightarrow$)} Assume there is a $\mathcal{F}$-valid (simple) path $\Pi = (s, v_1, \ldots, v_k, t)$ in $G$. Based on $\Pi$ we suggest a flow in $N$ where a unit of flow departing $s^{\prime}$ reaches the gadget of $v_1 = s$ while gaining $B$ units of flow.
The accumulated cost so far is $0$.
Based on the construction of every gadget, following the path $\Pi$, the $B+1$ units of flow can go through the corresponding gadget of every vertex $v_i$ and reach the gadget of $t$ with no loss and no cost.
Therefore, there is a feasible flow in $N$ from $s^{\prime}$ to $t^{\prime}$ with $0$ cost.

\noindent
\textbf{($\Leftarrow$)}
Let $f$ be a feasible flow in $N$ along a (simple) path $\Pi = (s^{\prime},s, \ldots,t, t^{\prime})$ with no cost, where seed flow $f_1 = 1$.
Based on the construction of $N$ from $G$, in a coarser view, a feasible flow goes from one gadget to another gadget.
Hence, path $\Pi$ can be viewed as $\Pi = (s^{\prime},\underline{s}, \underline{v_1}, \ldots, \underline{v_k}, \underline{t},t^{\prime})$,
where $\underline{v_i}$ represent the gadget corresponding to vertex $v_i$ that some of its edges are used in path $\Pi$.

The unit flow departs $s^{\prime}$ and $B+1$ units reach the gadget $\underline{s}$ with cost $0$.
When $B+1$ units reach the gadget $\underline{v}$ of $v$:
\begin{itemize}
  \item \textbf{If $v$ is not involved in any forbidden transition:} As explained in the third step of Procedure~\ref{proc:PAFT-to-Shortest},
  $B+1 - x$ units reach the gadget of one of the neighbors of $v$ with cost $xB$ for $0 \le x < 1$.
  \item \textbf{If $v$ is involved in some forbidden transitions:}
  \begin{itemize}
    \item If $v$ is an instance of $4$-(a) or $4$-(b) in Procedure~\ref{proc:PAFT-to-Shortest}, flow can reach the gadget of exactly one of neighbors of $v$ only if no forbidden transition
      is used.
    \item If $v$ is an instance of $4$-(c) in Procedure~\ref{proc:PAFT-to-Shortest}, as explained in this case, the flow
    is either fully consumed or suffers cost $2B(B+1)$, if any forbidden transition is used.
    The entering flow to this gadget reaches the neighboring gadget with no loss and no cost, only if no forbidden transition is used by the flow.
    \item For every gadget involved in some forbidden transitions, it is always the case that the entering $B+1$ units of flow
    either is fully consumed or suffers no loss upon reaching the neighboring gadget.
  \end{itemize}
\end{itemize}
Every feasible flow $f$ originated from $s^{\prime}$ with seed flow $f_1 = 1$ as reaches $t^{\prime}$
has gained $B-x$ units, with accumulated cost $xB$ for $0 \le x <1$, if no forbidden transition is used.
On the other hand, if any forbidden transition is used, a feasible flow along a path from $s^{\prime}$ to $t^{\prime}$ costs at least $2B(B+1)$, as explained in the sub-cases of case $4$ in Procedure~\ref{proc:PAFT-to-Shortest}.
Accordingly, a feasible flow with $f_1 = 1$ along some path from $s^{\prime}$ to $t^{\prime}$ has minimum cost
$0$, only if no forbidden transition is used and $x = 0$ (\emph{i.e.}, $B+1$ units reach $t^{\prime}$ with no loss).
\end{proof}
The following theorem concludes this section.
\begin{theorem}
\label{thr:shortestPath}
The simple shortest path problem is NP-hard in the strong sense for additive flow networks where the underlying graph is planar.
\end{theorem}
\begin{proof}
Proof is immediate based on Lemmas~\ref{lemma:PAFT} and~\ref{lemma:PAFT-to-network} and the fact that the procedure~\ref{proc:PAFT-to-Shortest} can be carried out in polynomial time (with respect to the size of the input graph $G$ and the set of forbidden transitions $\mathcal{F}$).
\end{proof}

\section{Maximum flow problem in additive flow network}
\label{sect:maxFlow}
  In~\cite{brandenburg2011shortest}, the authors study the problem of maximum flow with additive gains and losses for general graphs. They show that
finding maximum in-flow and out-flow are NP-hard tasks for general graphs.
In this section we show the same result for planar additive flow networks.
In order to show this result, we facilitate a special variation of planar satisfiability problem (planar SAT),
which is briefly defined in the following.
\begin{definition}{Strongly planar CNF} Conjunctive normal form (CNF for short) formula $\varphi = (\mathcal{X}, \mathcal{C})$ with the set of variable $\mathcal{X}$ and the set of clauses $\mathcal{C}$ is \emph{strongly planar} if graph $G_{\varphi} = (V, E)$ constructed as follows is planar:
\begin{enumerate}
  \item $V$ contains a vertex for every literal and one vertex for every clause.
  \item $\Set{x, \widetilde{x}} \in E$ for every $x \in \mathcal{X}$, where $\widetilde{x}$ denotes the negation of boolean variable $x$.
  \item If a claus $C$ contains literal $\overline{x}$ (which can be $x$ or $\widetilde{x}$), there is an edge $\Set{\overline{x}, C}$
  connecting vertex $\overline{x}$ and the vertex corresponding to $C$.
  \item No other edge exists other than those introduced in Part $2$ and Part $3$.
\end{enumerate}
\end{definition}
\begin{example}
\label{exp-planar-cnf}
Consider CNF formula $\varphi = (\widetilde{x} \vee y \vee \widetilde{z})\wedge (x \vee \widetilde{y} \vee z) \wedge (x \vee w \vee z) \wedge (\widetilde{x} \vee \widetilde{w} \vee \widetilde{z})$. Figure~\ref{fig:cnf} represents a planar embedding of the graph $G_{\varphi}$.
\begin{figure}[h]
    \begin{center}
        \includegraphics[scale=0.6]{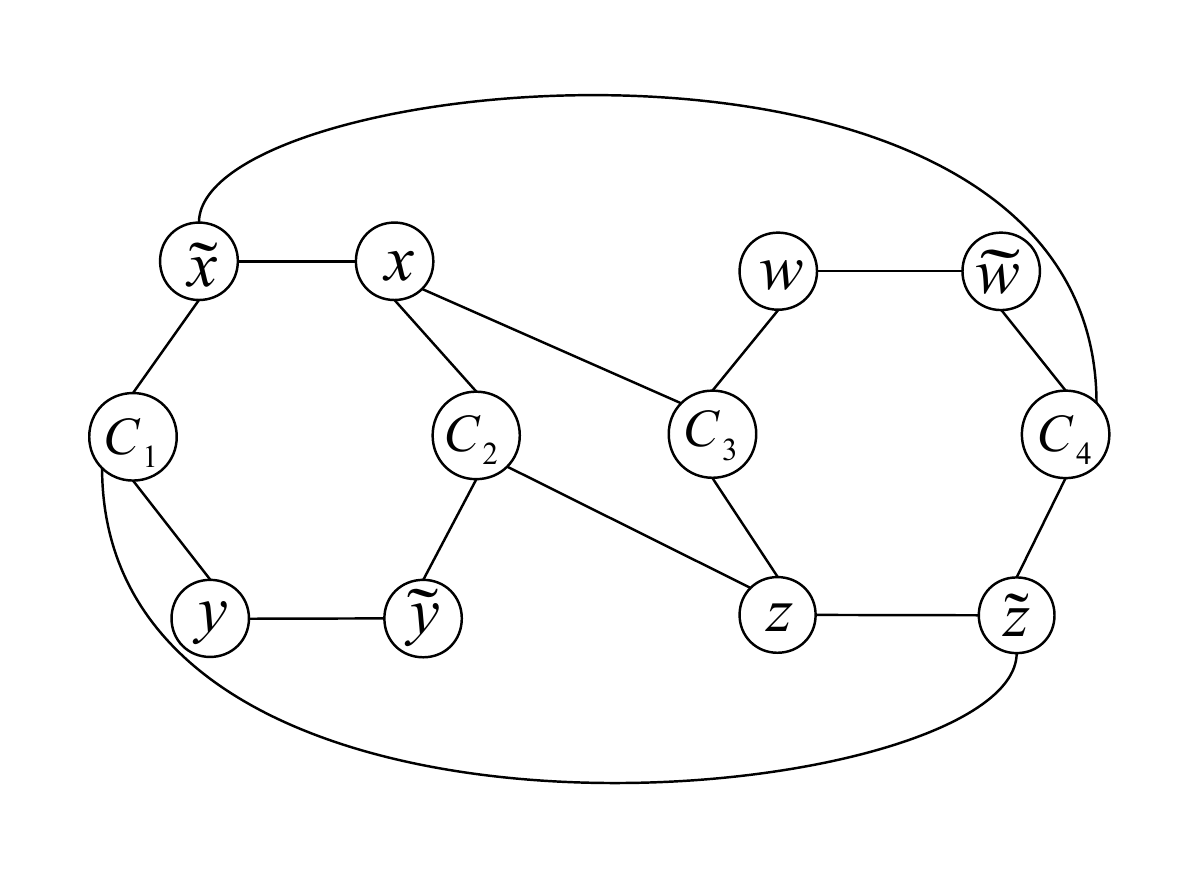}
    \end{center}
    \vspace{-0.4in}
    \caption{Underlying graph of the 3CNF formula $\varphi = (\widetilde{x} \vee y \vee \widetilde{z})\wedge (x \vee \widetilde{y} \vee z) \wedge (x \vee w \vee z) \wedge (\widetilde{x} \vee \widetilde{w} \vee \widetilde{z})$. }
    \label{fig:cnf}
    \vspace{-0.1in}
\end{figure}
\end{example}
\begin{definition}{1-in-3SAT} Given 3CNF formula $\varphi$, 1-in-3SAT is the problem of finding a satisfying assignment such that in each clause, exactly one of the three literals is assigned to $1$.
\end{definition}

\begin{lemma}
\label{lemma:1-in-3SAT}
Strongly Planar 1-in-3SAT is NP-complete.
\end{lemma}
For the proof of Lemma~\ref{lemma:1-in-3SAT}, we refer the reader to~\cite{wu2015strongly}.
We use Lemma~\ref{lemma:1-in-3SAT} to present the main contribution of this section (stated in Theorems~\ref{trm:additive-maxflow} and~\ref{thr:max-out-flow}).
Initially, Procedure~\ref{proc:planarSAT-to-flownet} shows three steps in order to transform a CNF formula $\varphi$
into an additive flow network $N_{\varphi}$.
\begin{procedure}{CNF into additive flow network}
\label{proc:planarSAT-to-flownet}
Consider graph $G_{\varphi} = (V, E)$ corresponding to a 3CNF $\varphi$ as an instance of 1-in-3SAT.
Starting from the undirected graph $G_{\varphi}$, additive flow network $N_{\varphi}$ is constructed using the following steps:
\begin{enumerate}
  \item Every undirected edge $\Set{\overline{x}, C}$ connecting  literal $\overline{x}$ (which can be $x$ or $\widetilde{x}$) to clause $C$ is replaced by a directed edge
$e = (x,C)$ from $x$ to $C$ where $u(e) = 1$ and $g(e) = 0$.
  \item For every clause-vertex $C_i$, a sink vertex $t_i$ and an edge $e = (C_i, t_i)$ are introduced, where $u(e) = 1$ and $g(e) = 0$.
  \item Every pair of literal-vertices $\Set{x, \widetilde{x}}$  is replaced by a gadget as shown in Figure~\ref{fig:gadget}.
\end{enumerate}
\begin{figure}[h]
    \begin{center}
        \includegraphics[scale=0.5]{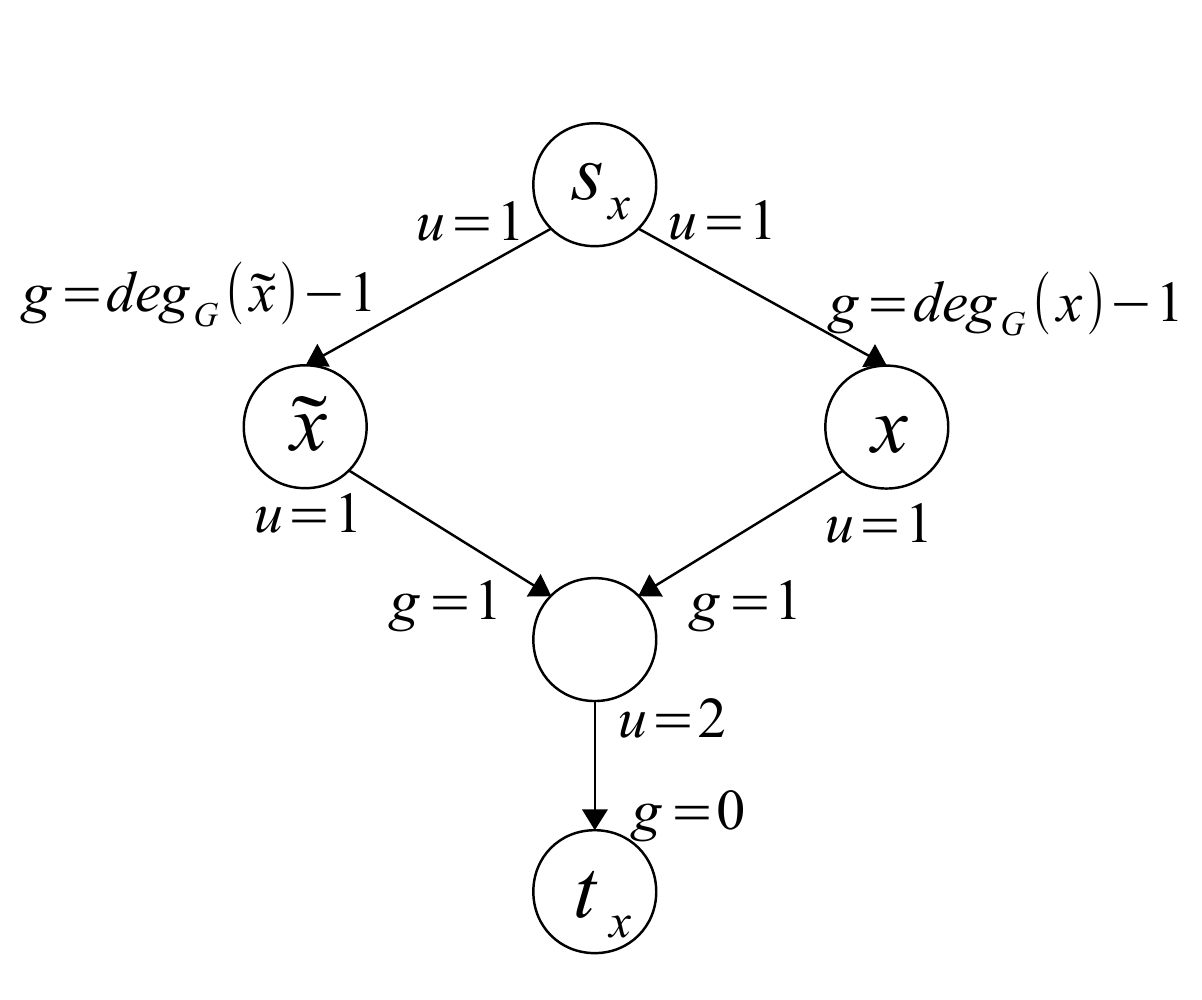}
    \end{center}
    \vspace{-0.3in}
    \caption{The gadget that replaces every pair of literal-vertices $\Set{x, \widetilde{x}}$. The upper bound and the gain or loss of every edge can be found by that edge, where $deg_{G}(x)$ is the degree of literal-vertex $x$ in the graph $G_{\varphi}$. The gadget of
    every pair $\Set{x,\widetilde{x}}$ introduces one source $s_x$ and one sink vertex $t_x$.}
    \label{fig:gadget}
    \vspace{-0.2in}
\end{figure}
\end{procedure}
\begin{example}
The flow network corresponding to the CNF formula $\varphi$ of example~\ref{exp-planar-cnf} is depicted in figure~\ref{fig:cnf-net}.
The missing capacity values and gain factors are $1$ and $0$, respectively. For simplicity some source vertices (also some sink vertices) are combined. It is easy to see that the graph of the constructed network $N_{\varphi}$ is planar iff $G_{\varphi}$ is planar.
\begin{figure}[H]
    \begin{center}
        \includegraphics[scale=0.5]{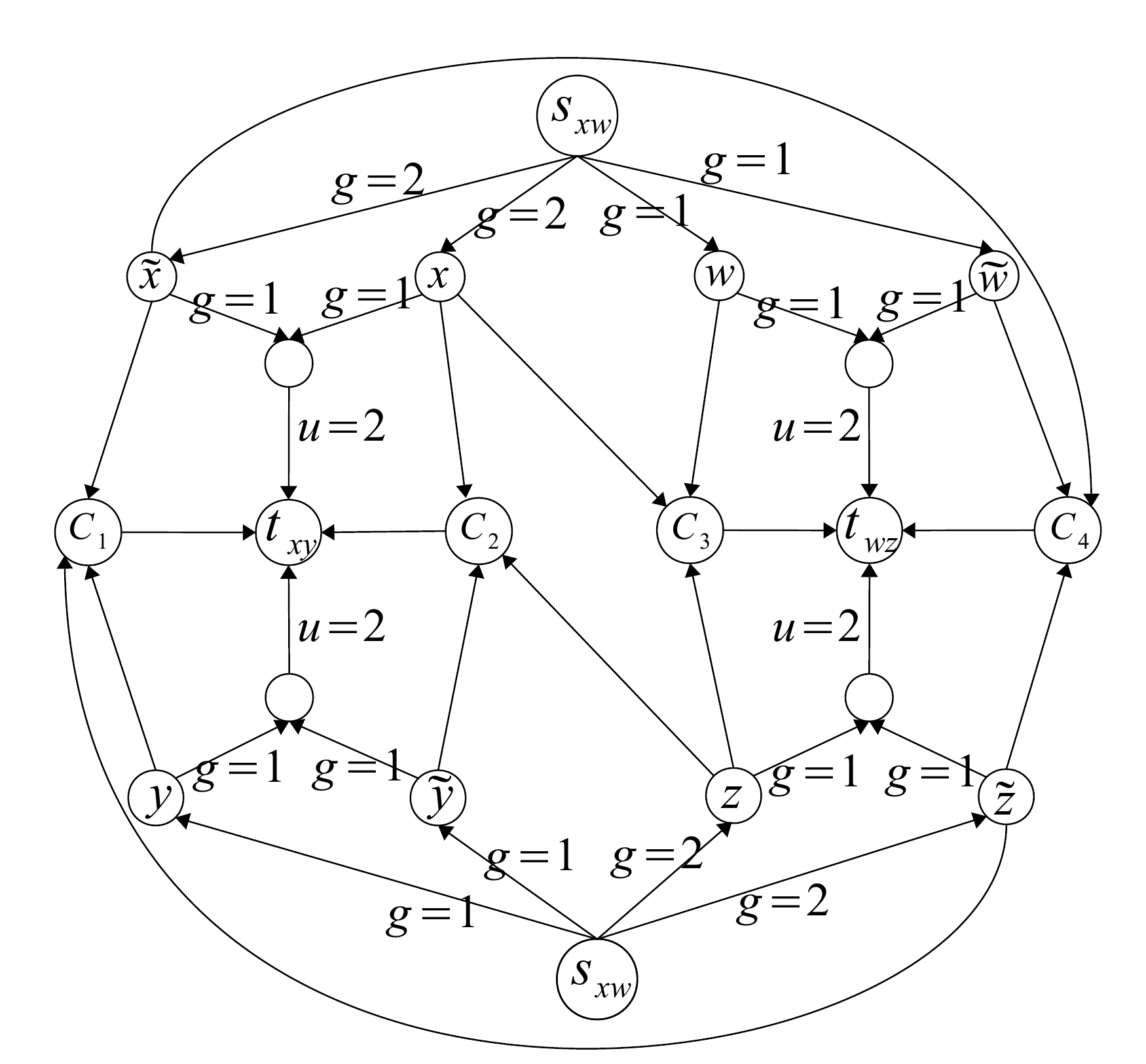}
    \end{center}
    \vspace{-0.2in}
    \caption{The flow network constructed from the graph $G_{\varphi}$ of CNF formula $\varphi = (\widetilde{x} \vee y \vee \widetilde{z})\wedge (x \vee \widetilde{y} \vee z) \wedge (x \vee w \vee z) \wedge (\widetilde{x} \vee \widetilde{w} \vee \widetilde{z})$.}
    \label{fig:cnf-net}
    \vspace{-0.2in}
\end{figure}
\end{example}
\begin{lemma}
\label{lemma:SAF-iff-Maxflow}
CNF formula $\varphi = (\mathcal{X}, \mathcal{C})$ is a positive instance of strongly planar 1-in-3SAT iff
the maximum in-flow in additive flow network $N_{\varphi}$ is $2|\mathcal{X}|+|\mathcal{C}|$, where $N_{\varphi}$ is constructed according to procedure~\ref{proc:planarSAT-to-flownet}.
\end{lemma}
\begin{proof}

\noindent
\textbf{($\Rightarrow$)} Consider a satisfying assignment of variables that every clause has exactly one literal with value $1$.
For every literal $\overline{x}$ assigned to $1$ we push $1$ unit flow through the edge connecting a source vertex to the related literal-vertex of $\overline{x}$. The amount of flow that reaches $\overline{x}$ is exactly one unit more than the number of clauses containing $\overline{x}$. This flow saturates all the edge connecting $\overline{x}$ to the vertices of the clauses that contain the literal $\overline{x}$. The remaining one unit flow reaches the sink vertex in the gadget containing the vertex of $\overline{x}$ (plus an extra unit flow gained). Every edge connecting a source to the vertices corresponding to the literals that are assigned $0$, will not be used. It is easy to check that the suggested flow is feasible and all the edges connected to sink vertices are saturated.

\noindent
\textbf{($\Leftarrow$)} We produce a variable assignment for $\varphi$ based on a given feasible $f$ for $N_{\varphi}$. According to the construction of the gadgets for every pair of literals $\Set{x, \widetilde{x}}$ in Figure~\ref{fig:gadget}, in every feasible flow at most one of the two edges $(s_x,x)$ and $(s_x,\widetilde{x})$ can be used (\emph{i.e.} in every feasible flow $f$, for every $x \in X$, $f(x) = 0 \text{ or } f(\widetilde{x}) = 0$). A literal $\overline{x}$ is set to $1$ if $f((s_x, \overline{x})) > 0$, and is set to $0$ otherwise.

Hence, given $N_{\varphi}$ where $f_{in} = 2|\mathcal{X}|+|\mathcal{C}|$, all the edges reaching a sink vertex are saturated. Namely: (i) every edge connecting a clause-vertex to a sink vertex is saturated, which make up $|\mathcal{C}|$ units of flow (\emph{i.e.} corresponding clause is satisfied) and (ii) the remaining $2|\mathcal{X}|$ units of flow is supplied by the sink vertices of all the gadgets (\emph{i.e.} $f((s_x, x)) = 1$ or $f((s_x, \widetilde{x})) = 1$ for every pair $\Set{x, \widetilde{x}}$). Accordingly, if a feasible flow has maximum in-flow $2|\mathcal{X}| + |\mathcal{C}|$, it is the case that for every variable $x$ only one of the literal-vertices $x$ or $\widetilde{x}$ has entering flow, hence in the suggested assignment that literal is set to $1$. Also for every clause $C_i$ one unit flow reaches its corresponding vertex, which means that clause is satisfied and only one of its literal is set to $1$.
\end{proof}
It is not hard to check that Procedure~\ref{proc:planarSAT-to-flownet} can be done in polynomial time in the size of the input CNF formula.
Hence, based on Lemmas~\ref{lemma:1-in-3SAT} and~\ref{lemma:SAF-iff-Maxflow}, the following theorem can be deduced immediately.
\begin{theorem}
\label{trm:additive-maxflow}
computing a feasible flow $f$ with maximum in-flow $f_{in}$ is an NP-hard problem in the strong for the class planar additive flow networks.
\end{theorem}
In the context of max-flow problem,\footnote{Note that in this context cost functions are irrelevant.} for every additive flow network $N = (V, E,S, T, u, g)$ there exists a reversed flow network $N^{\prime} = (V^{\prime}, E^{\prime}, S^{\prime}, T^{\prime}, u^{\prime}, g^{\prime})$, constructed by reversing the direction of every edge and swapping source vertices with sink vertices
(\emph{i.e.} $S^{\prime} = T$ and $T^{\prime} = S$). Given edge $e = (v, u)$ in $E$, we have $e^{\prime} = (u, v) \in E^{\prime}$ where $u^{\prime}(e^{\prime}) = u(e) + g(e)$ and $g(e^{\prime}) = - g(e)$. Hence, if edge $e$ is gainy (lossy) in $N$, $e^{\prime}$ is lossy (gainy) in $N^{\prime}$.

Based on the definition of reversed flow networks, the following lemma is straightforward and the details can be found in~\cite{brandenburg2011shortest}.
\begin{lemma}
\label{lemma:rev-flow-net}
Consider the feasible flow $f$ for an additive flow network $N$, where for every edge $e, g(e) \ge 0$ (in other words, there is no lossy edge in $N$).
Then exists a flow $f^{\prime}$ for the reversed network $N^{\prime}$, where $f_{out}^{\prime} = f_{in}$ (and similarly $f_{in}^{\prime} = f_{out}$).
\end{lemma}
In the construction of $N_{\varphi}$ from $G_{\varphi}$ based on Procedure~\ref{proc:planarSAT-to-flownet}, there is no lossy edge. Hence the following theorem, as an immediate result of lemma~\ref{lemma:rev-flow-net}, concludes this section. The proof is straightforward and left to the reader.
\begin{theorem}
\label{thr:max-out-flow}
In planar additive flow networks, finding feasible flow $f$ with maximum out-flow $f_{out}$ is an NP-hard problem in the strong sense.
\end{theorem}

\section{Conclusion and future work}
\label{sect:future}
  In this report we investigated the max-flow and shortest path problems
for flow networks with additive gains and losts when the underlying
graph is planar.  In Sections~\ref{sect:shortestPath}
and~\ref{sect:maxFlow}, we show that both problems are NP-hard in
the strong sense for planar additive flow networks, \emph{i.e.} even when
all the values of cost, gain and capacity functions assigned to every
edge are bounded by polynomials in the size of the input network.

Hence, there is the question of existence of approximation algorithms
for any of the two problems.  To our best knowledge, no approximation
algorithm has yet been suggested for any of those problems for
additive flow networks (with or without any restriction on the
structure of the underlying graph).

The other question to investigate is the existence of polynomial time
algorithms when some input parameters are fixed. For instance, based
on the notion of outerplanarity, every planar graph is $k$-outerplanar
for some integer $k \ge 1$.  In~\cite{BestKfoury:dsl11,
SouleBestKfouryLapets:eoolt11,kfoury2013different} the authors
introduce and study a compositional framework for the analysis of flow
networks (based on a so-called \emph{Theory of Network Typings}); 
based on this framework in~\cite{kfoury2013compositional}, a linear time
algorithm (with respect to the number of vertices) for max-flow
problem in $k$-outerplanar graphs is suggested, when $k$ is fixed.  We
believe that, with some minor modifications in the suggested
framework, the same result can be achieved for max-flow problems in
additive flow networks.  The Theory of Network Typings
proposes an algebraic approach for flow networks that allows a
compositional analysis of flow based on polyhedral
computations.  As defined so far, this framework does not account
for the presence of cost functions on the flow.  Hence, another problem
left for future investigation is the problem of incorporating cost
functions in that framework thereby allowing a compositional
analysis of the shortest path problem in additive flow networks.


\Hide
{\footnotesize
\printbibliography
}

{\footnotesize 
\bibliographystyle{plainurl} 
\bibliography{bibs}
}

\ifTR
\else
\fi

\end{document}